\documentclass[journal]{IEEEtran}

\usepackage{pifont,amssymb,url}
\usepackage{multicol,mathrsfs}
\usepackage[usenames,dvipsnames,svgnames,table]{xcolor}
\usepackage{bm,bbm,mathrsfs,amssymb,pifont,amssymb,pifont,amsfonts,amsmath,stmaryrd,color,amssymb,graphics,graphicx,epstopdf}
\usepackage{amscd,graphicx,array,dsfont,texdraw,tikz,enumerate,multicol,multirow,verbatim}
\usetikzlibrary{arrows,shapes,shadows,positioning,automata,patterns}
\usetikzlibrary{matrix,arrows,calc,trees,decorations.pathmorphing,decorations.markings, decorations.pathreplacing}
\usepackage[utf8x]{inputenc}
\usepackage[T1]{fontenc}
\usepackage{booktabs}
\usepackage{algorithm,algorithmicx,algpseudocode}

\usepackage{bm,amsmath,amssymb,amsthm,graphicx,subfigure,epstopdf}
\usepackage[comma,numbers,square,sort&compress]{natbib}
\usepackage{epstopdf,epic,subfigure,epsfig,arydshln,enumerate}
\usepackage{cuted} 
\usepackage{fancyhdr,lastpage,color,dsfont}

\newcommand{\myr}{ }
\usetikzlibrary{positioning}

\newcommand{\col}{\mbox{col}}

\newtheorem{rem}{Remark}
\newtheorem{Corollary}{Corollary}

\newtheorem{prob}{Problem}

\newtheorem{thm}{Theorem}
\newtheorem{lem}{Lemma}

\newtheorem{ass}{Assumption}

\allowdisplaybreaks
\begin{document}
\title{Nonadaptive Output Regulation of Second-Order Nonlinear  Uncertain Systems}
\author{Maobin~Lu, Martin~Guay, Telema~Harry,  Shimin~Wang, Jordan~Cooper
\thanks{
This work was supported in part by the Natural Sciences and Engineering Research Council (NSERC), Canada,
in part by the National Science and Technology Major Project of China under Grant 2021ZD0112600,
in part by the National Natural Science Foundation of China under Grant
62373058,
and in part by the Beijing Natural Science Foundation under Grant L233003. 
Maobin Lu is with the School of Automation, Beijing Institute of Technology, Beijing 100081, China, and is also with the Beijing Institute of Technology Chongqing Innovation Center, Chongqing 401120, China. 
Martin~Guay, Telema~Harry,  and Jordan~Cooper are with the Department of Chemical Engineering, Queen's University, Kingston,  ON K7L 3N6, Canada (martin.guay@queensu.ca, telema.harry@queensu.ca, shimin.wang@queensu.ca) 
Shimin~Wang is with  Massachusetts Institute of Technology, Cambridge,  MA 02142, USA.\\
(Corresponding author: Shimin Wang)

}
}

\maketitle 

\begin{abstract}\myr 
This paper investigates the robust output regulation problem of second-order nonlinear uncertain systems with an unknown exosystem. 
Instead of the adaptive control approach, this paper resorts to a robust control methodology to solve the problem and thus avoid the bursting phenomenon.
In particular, this paper constructs  
generic internal models for the steady-state state and input variables of the system.
By introducing a coordinate transformation, this paper converts the robust output regulation problem into a nonadaptive stabilization problem of an augmented system composed of the second-order nonlinear uncertain system and the generic internal models. 
Then, we design the stabilization control law and construct a strict Lyapunov function that guarantees the robustness with respect to unmodeled disturbances. 
The analysis shows that the output zeroing manifold of the augmented system can be made attractive by the proposed nonadaptive control law, which  
solves the robust output regulation problem. Finally, we demonstrate the effectiveness of the proposed nonadaptive internal model approach by its application to the control of the Duffing system.

\end{abstract}

\section{Introduction}
 
The output regulation problem, traditionally called the servomechanism problem, is one of the central problems in systems and control. 
It aims to design a control law for a controlled system such that asymptotic tracking for time-varying reference trajectory is achieved while rejecting time-varying disturbances \cite{isidori1990output,huang2004nonlinear}. The dynamics of the reference trajectory and the disturbances are generally governed by an autonomous system called the exosystem. 
The output regulation problem addresses a stabilization problem, disturbance rejection problem and trajectory tracking problem as special cases. Consequently, the scope of the approach is broad and  
it can be used to solve many practical problems, such as flying control of unmanned aerial vehicles, manipulation of robot arms, and motion control of mobile vehicles \cite{isidori2003robust,huang2004nonlinear,bin2020approximate,broucke2022adaptive,su2021adaptive}. 


It is known that the internal model principle plays a crucial role in output regulation. 
For linear systems, the internal model principle converts the output regulation problem into a pole assignment problem \cite{francis1976internal,francis1977linear,davison1976robust}. 
%
For nonlinear systems, the internal model design and the output regulation problem become more challenging \cite{isidori1990output,huang1990nonlinear}.
It was shown that the linear internal model for output regulation of linear systems cannot work for nonlinear output regulation \cite{huang1994robust}. 
The reason is that for nonlinear output regulation, the steady-state behaviour of nonlinear systems is described by nonlinear functions of the exogenous signal generated by the exosystem.
Moreover, unlike linear output regulation, the nonlinear output regulation problem requires a system specific approach that is not amenable to a general solution methodology. Various internal models have been developed to solve the nonlinear output regulation problem \cite{isidori1990output,huang1995asymptotic,isidori2003robust,huang2004general}. 
Nonlinear internal models were proposed in \cite{huang2004general,byrnes2003limit,byrnes2004nonlinear} to deal with the nonlinear output regulation problem. 
In particular, it was shown in 
\cite{huang2001remarks} that the nonlinear output regulation problem is solvable if the steady-state input is a trigonometric polynomial of the exogenous signals.
This condition is equivalent to the partial differential condition stated in \cite{byrnes1997structurally} when the exogenous signal is a sum of finitely many harmonics. 
A general framework was established in \cite{huang2004general}, which converts the output regulation problem into a stabilization problem of an augmented system composed of the controlled nonlinear system and the internal model. 

In situations where the exosystem contains uncertainties, some adaptive internal models have been proposed for the solution of the output regulation problem. Adopting the tools and methods from the adaptive control literature, this class of output regulation problems is also called the adaptive output regulation problem. 
In particular, pioneered in \cite{nikiforov1998adaptive}, the canonical internal model has been used to solve the output regulation problem for various linear systems \cite{marino2003output,basturk2015adaptive} and for nonlinear systems \cite{serrani2001semi,liu2009parameter,li2012nonlinear,lu2019adaptive}. 
Using these adaptive internal models, disturbances in the form of a sum of finitely many harmonics with unknown frequencies, magnitude and initial phases can be effectively tackled in the output regulation problem. 
Moreover, as shown in \cite{liu2009parameter}, the estimated parameters would converge to the real values when the canonical internal model is a minimal internal model. 
In adaptive output regulation designs,  a Lyapunov function is usually constructed with a negative semi-definite derivative along the trajectories of the closed-loop system.
Using standard adaptive control arguments, Barbalat's lemma is generally used to demonstrate the system's stability which is not established in Lyapunov sense. 
Some nonadaptive methods have been developed for the output regulation problem subject to unknown exosystems \cite{marconi2008uniform,isidori2012robust}. 
%
%
%
Some generic internal models have also been developed to solve the output regulation problem subject to an uncertain exosystem \cite{isidori2012robust,wang2023nonparametric}. 
%
%
Recently, a generic internal model was constructed in \cite{wang2023nonparametric} to solve the output regulation problem of a class of nonlinear systems in the output feedback form. 
A nonparametric learning method was proposed in \cite{wang2023nonparametric} where the steady-state input can be polynomials of the exogenous signals. Subsequently, this method was extended to address more general nonlinear cases in \cite{wang2024nonparametric}.
{\myr In addition, the results in \cite{wang2023nonparametric,wang2024nonparametric} enable the steady-state generator to be polynomial in the exogenous signal for nonlinear output regulation while relaxing restrictive assumptions on the exosystem, such as the requirement for an even dimension and the absence of zero eigenvalue in existing results. 
}

In this paper, we study the global robust output regulation problem of a class of second-order nonlinear systems subject to an unknown exosystem. 
It is known that the practical physical system can be modelled as a second-order system by Newton's law, and thus, the investigation under consideration is of practical and technical significance. 
We design a class of generic internal models, the output mapping of which is independent of the uncertain parameters of the exosystem. 
Following the general framework in \cite{huang2004general}, we introduce some coordinate transformation and convert the robust output regulation problem into a stabilization problem of the augmented system composed of the generic internal model and the second-order nonlinear system. 
Then, we design a nonadaptive stabilizing control law for the augmented system and construct the strict Lyapunov function. 
By Lyapunov analysis, we show that the augmented system is globally asymptotically stable, and thus, the robust output regulation problem is solved in spite of the unknown exosystem. Compared with the existing adaptive internal model approach to the nonlinear output regulation problem with unknown exosystems, the proposed generic internal model approach has some advantageous features. First, the inverse dynamics associated with the internal model in the augmented system have some input-to-state stability property, facilitating the nonadaptive stabilization control law design. Second, the adaptive dynamics for compensation of the unknown parameters in \cite{liu2009parameter} and \cite{lu2019adaptive} caused by the unknown exosystem are not necessary for the proposed robust stabilization control approach.
Third, a strict Lyapunov function is constructed and thus the system is stable in the sense of Lyapunov.
{\myr This results in improved robustness to unmodeled disturbances for the proposed generic internal model approach in comparison to the adaptive internal model approach.}

The rest of this paper is organized as follows. In Section \ref{section2}, we introduce some standard assumptions and lemmas. Section \ref{mainresults} is devoted to presenting the main results. This is followed by a simulation example in Section \ref{numerexam} and brief conclusions in Section~\ref{conlu}.

\textbf{Notation:} $\|\cdot\|$ is the Euclidean norm. $\emph{Id} : \mathds{R}\rightarrow \mathds{R}$ is an identity function. For $X_1,\dots,X_N\in \mathds{R}^{n}$, let $\col(X_1,\dots,X_N)=[X_{1}^{T},\dots ,X^{T}_{N}]^T$. 
A function $\alpha: \mathds{R}_{\geq 0}\rightarrow \mathds{R}_{\geq 0}$ is of class $\mathcal{K}$ if it is continuous, positive definite, and strictly increasing. $\mathcal{K}_o$ and $\mathcal{K}_{\infty}$ are the subclasses of bounded and unbounded $\mathcal{K}$ functions, respectively. For functions $f_1(\cdot)$ and $f_2(\cdot)$ with compatible dimensions, their composition $f_{1}\left(f_2(\cdot)\right)$ is denoted by $f_1\circ f_2(\cdot)$. For a matrix $X$,  $\textnormal{\mbox{Adj}}[X]$ denotes its adjugate matrix.

\section{Problem Formulation and Assumptions}\label{section2}

In this paper, we consider second-order nonlinear uncertain systems of the  following form:
\begin{subequations}\label{second-nonlinear-systems}
\begin{align}
\dot{x}_1&= x_2\\
 \dot{x}_2 &=f(x_1, x_2, v, w)+b(v,w)u \\
 y&=x_1
 \end{align}
\end{subequations}
where $x=\col(x_1,x_2)\in \mathds{R}^{2}$ is the vector of state variables, $y\in \mathds{R}$ denotes the output of the system, $u\in \mathds{R}$ is the input, $w\in \mathds{W}\subset \mathds{R}^{n_w}$ is an uncertain parameter vector with $\mathds{W}$ being an arbitrarily prescribed subset of $\mathds{R}^{n_w}$ containing the origin, and $v(t)\in \mathds{R}^{n_v}$ is an exogenous signal representing the reference input and disturbance, which is generated by the  exosystem as follows:
\begin{subequations}\label{eqn: exosystem system}\begin{align}
\dot{v}&=S(\sigma)v\\ 
 e &= y - h(v,w)
\end{align}\end{subequations}
where $S(\sigma)\in \mathds{R}^{n_v \times n_v}$ is a constant matrix, $e$ is the tracking error, and $\sigma\in \mathds{R}^{n_\sigma}$ is an uncertain parameter vector in a compact set $\mathds{S}  \subset \mathds{R}^{n_\sigma}$.
The functions $f(\cdot)$, $b(\cdot)$ and $h(\cdot)$  are globally defined and sufficiently smooth satisfying $f(0,0,0,w)=0$ and $h(0,w)=0$ for all $w\in \mathds{W}$.

The nonlinear robust output regulation problem is formulated as follows.

\begin{prob} \label{Prob: second-order-Output-regulation}
Given the nonlinear system composed of \eqref{second-nonlinear-systems} and \eqref{eqn: exosystem system}, any compact subsets $\mathds{S}\in \mathds{R}^{n_{\sigma}}$, $\mathds{W}\in \mathds{R}^{n_w}$ and $\mathds{V}\in \mathds{R}^{n_v}$ with $\mathds{W}$ and $\mathds{V}$ containing the origin, design a control law of the form
\begin{align}\label{Procontr}
u &=  \varphi(\vartheta,e,x)\notag\\
\dot{\vartheta} &=  \phi(\vartheta,e,x)
\end{align}
where $\vartheta\in \mathds{R}^{n_\vartheta}$, $\varphi(\cdot)$ and $\phi(\cdot)$ are sufficiently smooth functions vanishing at the origin, such that for $\sigma \in \mathds{S}$, $v(0)\in \mathds{V}$ and $w\in \mathds{W}$, and any initial states $x(0)$ and $\vartheta(0)$, the solution of the closed-loop system composed of \eqref{second-nonlinear-systems} and \eqref{Procontr} exists and is bounded for all $t\geq 0$,
 and $$\lim\limits_{t\rightarrow\infty}e(t)=0.$$
\end{prob}

To solve the problem,  we need to compensate for the steady state of the system
which requires the solution of the regulator equations associated with the composite system \eqref{second-nonlinear-systems} and \eqref{eqn: exosystem system}. These equations are given as follows:
\begin{subequations}\label{regeq1}
\begin{align} 
    \mathbf{x}_1(v(t), w, \sigma)  &= h(v(t), w)\\
    \mathbf{x}_2(v(t), w, \sigma) &= \frac{\partial h(v(t), w)}{\partial v(t)} S(\sigma)v(t) \\
    \mathbf{u}(v(t),w,\sigma)  &= b( v(t), w)^{-1} \Big( \frac{\partial \mathbf{x}_2(v(t), w, \sigma)}{\partial v(t)} S(\sigma)v(t)\notag \\
    & \hspace{-0.25in} - f(\mathbf{x}_1(v(t),w,\sigma), \mathbf{x}_2(v(t),w,\sigma),v(t),w)\Big)
\end{align}
\end{subequations}
where $\mathbf{x}(v(t),w,\sigma)= \textnormal{col} (\mathbf{x}_1(v(t),w,\sigma), \mathbf{x}_2(v(t),w,\sigma))$ and $\mathbf{u}(v(t),w,\sigma)$ are referred to as the steady-state state variables and the steady-state input, respectively. 

By \eqref{regeq1}, the steady-state state variables, and the steady-state input contains uncertainties and unknown disturbances. Consequently, they are not directly available.
To achieve exact compensation for output regulation, some standard assumptions are listed as follows.

\begin{ass}\label{ass0} For all $\sigma \in \mathds{S}$, all the eigenvalues of $S(\sigma)$ are simple
with zero real parts.
\end{ass}

\begin{ass}\label{ass0i} The function $\mathbf{u}(v(t), w, \sigma)$ and $\mathbf{x}_2(v(t), w, \sigma)$ are
polynomials in $v(t)$ with coefficients depending on $w$ and $\sigma$.
\end{ass}
 \begin{rem}\label{remPE}
To facilitate later analysis, let 
\begin{align*}\bm{u}_1(v(t), \sigma, w)=&\  \mathbf{x}_2(v(t), w, \sigma),\\
    \bm{u}_2(v(t),\sigma,w)=&\  \mathbf{u}(v(t), w, \sigma).\end{align*}
From \cite{huang2001remarks,liu2009parameter}, under Assumption \ref{ass0i}, there exist integers $n^*_1>0$ and $n^*_2>0$ such that 
 $\bm{u}_i(v(t),\sigma,w)$ can be expressed by
\begin{align}
\bm{u}_i(v(t),\sigma,w)=&\,\sum_{j=1}^{n_i^*}C_{i,j}(v(0), w,\sigma)e^{\imath \hat{\omega}_{i,j}t}
\end{align}
for some functions $C_{i,j}(v(0), w,\sigma)$, where $\imath$ is the imaginary unit and $\hat{\omega}_{i,j}$ are distinct real numbers for {\myr $1\leq j \leq n^*_i$} and each $i=1,2$.
\end{rem}
\begin{ass}\label{ass5-explicit}  
 For any $v(0)\in \mathds{V}$, $w\in \mathds{W}$ and $\sigma\in \mathds{S}$, $C_{i,j}(v(0), w,\sigma)\neq 0$,  {\myr $1\leq j \leq n^*_i$} and $i=1,2$ .
\end{ass}

\section{Main results}\label{mainresults}

\subsection{Internal model design}

As shown in \cite{huang2004nonlinear}, under Assumptions \ref{ass0} and \ref{ass0i}, there exist positive integers $n_i$, such that $\bm{u}_i(v,\sigma,w)$ satisfy for all $\col(v,w,\sigma)\in \mathds{V}\times\mathds{W}\times \mathds{S}$,
{\myr \begin{align} \label{aode-explicit}
\frac{d^{n_i}\bm{u}_i(v,\sigma,w)}{dt^{n_i}}&+a_{i,1}(\sigma)\bm{u}_i(v,\sigma,w)+\dots\notag\\
&+a_{i,n_i}(\sigma)\frac{d^{n_i-1}\bm{u}_i(v,\sigma,w)}{dt^{n_i-1}}=0,\;\;\;\;i=1,2
\end{align}}
where $a_{i,1}(\sigma),\dots,a_{i,n_i}(\sigma)$ all belong to $\mathds{R}$.
The differential equation \eqref{aode-explicit} is characterized by the following polynomial $$\varsigma^{n_i}+a_{i,1}(\sigma)+a_{i,2}(\sigma)\varsigma+\dots+a_{i,n_i}(\sigma)\varsigma^{n_i-1},\;\;\;\;i=1,2,$$ which admits distinct roots with zeros real parts for all $\sigma\in \mathds{S}$. 

For $i=1,2$, let $a_i(\sigma)=\col(a_{i,1}(\sigma), \dots, a_{i,n_i}(\sigma))$ and $$\bm{\xi}_i(v,\sigma,w)= \col\!\left(\bm{u}_i(v,\sigma,w),\dots,\frac{d^{n_i-1}\bm{u}_i(v,\sigma,w)}{dt^{n_i-1}}\right).$$ For ease of presentation, we adopt the following slight abuse of notation, $$\bm{\xi}_i:=\bm{\xi}_i(v,\sigma,w),\;\;\;\;i=1,2.$$ 
Next, we define the following matrix-valued function,
\begin{align*}
  \Phi_i(a_i(\sigma)) =&\,\left[
                      \begin{array}{c|c}
                        0_{(n_i-1)\times 1} & I_{n_i-1} \\
                        \hline
                        -a_{i,1}(\sigma) &-a_{i,2}(\sigma),\dots,-a_{i,n_i}(\sigma) \\
                      \end{array}
                    \right], 
                    \end{align*}
and, row vector,                     
                    \begin{align*}
  \Gamma_i =&\,\left[
                \begin{array}{cccc}
                  1 & 0 & \dots &0 \\
                \end{array}
              \right]_{1\times n_i},\;\;\;\;i=1,2.
\end{align*}
As a consequence, it follows that $ \bm{\xi}_i\left(v,\sigma,w\right)$, $\Phi_i(a_i(\sigma))$ and $\Gamma_i$ satisfy the equations:
\begin{subequations}\label{stagerator}\begin{align}
\dot{\bm{\xi}}_i &= \Phi_i (a_i(\sigma))  \bm{\xi}_i,&\\
   \bm{u}_i &=  \Gamma_i \bm{\xi}_i, &i=1,2.
\end{align}\end{subequations}
System \eqref{stagerator} is called a steady-state generator with output $u_i$ as it can be used to produce the steady-state signal $\textbf{u}_i(v,\sigma,w)$\cite{huang2004nonlinear}. We define the controllable matrix pair $(M_i, N_i)$ as follows:
\begin{subequations}\label{MNINter}\begin{align}
M_i=&\,\left[
                      \begin{array}{c|c}
                        0_{(2n_i-1)\times 1} & I_{2n_i-1} \\
                        \hline
                        -m_{i,1} &-m_{i,2},\dots,-m_{i,2n_i} \\
                      \end{array}
                    \right],\\
N_i=&\,\left[
                \begin{array}{cccc}
                  0 & 0 & \dots &1 \\
                \end{array}
              \right]_{1\times 2n_i}^T,\;\;\;\;i=1,2,                  
\end{align}
\end{subequations}
where $m_{i,1}$, $m_{i,2}$ , $\dots$, $m_{i,2n_i}$ are chosen such that $M_i$ is Hurwitz. We also consider the following matrix-valued function
$$\Xi_i(a_i) =:\Phi_i(a_i)^{2n_i }+\sum\nolimits_{j=1}^{2n_i}m_{i,j}\Phi_i(a_i)^{j-1} \in \mathds{R}^{n_i \times n_i }.$$ 
For $i=1,2$, by \cite{wang2023nonparametric,wang2024nonparametric}, it follows that, under Assumptions \ref{ass0}-\ref{ass5-explicit}, the matrix $\Xi_i(a_i) $ is non-singular and 
\begin{align}\label{XIQA-explicit}
\Xi_i (a_i)^{-1}=\textnormal{\col}\left(Q_{i,1}(a_i),\dots,Q_{i,n_i}(a_i)\right)\in \mathds{R}^{n_i \times n_i }
\end{align}
where  $$ Q_{i,j}(a_i)=\Gamma_i \Xi_i(a_i)^{-1}\Phi_i(a_i )^{j-1} \in \mathds{R}^{1\times n_i},\;\;j=1,\dots, n_i.$$
For $i=1,2$, we define the following Hankel real matrix \cite{afri2016state}: 
\begin{align}
\Theta_i (\theta_i)&=:\left[\begin{matrix}\theta_{i,1} &\theta_{i,2}&\dots&\theta_{i,n_i}\\
\theta_{i,2}&\theta_{i,3}&\dots&\theta_{i,n_i+1}\\
\vdots&\vdots&\ddots&\vdots\\
\theta_{i,n_i} &\theta_{i,n_i +1}&\dots&\theta_{i,2n_i -1}
\end{matrix}\right]\in \mathds{R}^{n_i \times n_i }\nonumber
\end{align}
where $$\theta_i = \textnormal{\col}(\theta_{i,1}, \theta_{i,2}, \dots, \theta_{i,2n_i})= Q_i \bm{\xi}_i(v(t), \sigma, w), $$  with 
\begin{align}\label{Qdefini}
Q_i=:\textnormal{\col}\left(Q_{i,1},\dots,Q_{i,2n_i}\right)\in \mathds{R}^{2n_i\times n_i}
\end{align}
and $$Q_{i,j}(a_i)=\Gamma_i \Xi_i(a_i)^{-1}\Phi_i(a_i)^{j-1} \in \mathds{R}^{1\times n_i},\;\;1\leq j\leq 2n_i.$$
The matrices $Q_i$, $M_i$, $N_i$,  $\Phi_i(a_i)$ and $\Gamma_i$ satisfy the following generalized Sylvester matrix equation:
 \begin{align}
M_i Q_i =&\,Q_i \Phi_i(a_i(\sigma))-N_i\Gamma_i, \;i=1,2.\label{MNGAMMAPhi}
\end{align}   
The existence of a matrix solution has been established in \cite{wang2023nonparametric,wang2024nonparametric,zhou2006new}. 
As shown in \cite{marconi2007output,marconi2008uniform,KE2003},  there exist nonlinear mappings 
\begin{align*}
\bm{\eta}^{\star}_i(v(t),\sigma,w)=&\,\int_{-\infty}^{t} e^{M_i(t-\tau)}N_i\bm{u}_i(v(\tau), \sigma, w) d\tau\\
\bm{u}_i(v(t), \sigma, w)=&\,\chi_i(\bm{\eta}^{\star}_i(v(t),\sigma,w)), ~~ \bm{\eta}^{\star}_i\in\mathds{R}^{2n_{i}},\;\;i=1,2,
\end{align*}
that satisfy the following differential equations 
\begin{align}\label{IM-01}
\frac{d \bm{\eta}^{\star}_i(v(t),\sigma,w)}{dt}  &= M_i\bm{\eta}_i^{\star}(v(t),\sigma,w) + N_i\bm{u}_i(v(t), \sigma, w) \notag\\
\bm{u}_i(v(t), \sigma, w) &= \chi_i(\bm{\eta}^{\star}_i(v(t),\sigma,w))
\end{align}
where 
$\chi_i(\cdot)$ is a sufficiently smooth function \cite{marconi2008uniform, KE2003,wang2023nonparametric}. {\myr As a result, equations \eqref{IM-01} construct generic internal models for both the steady-state state and input variables of the regulator equation \eqref{regeq1}. }
By using the matrix equation \eqref{MNGAMMAPhi}, it can be shown that system \eqref{IM-01} is also a steady-state generator of $\bm{u}_i(v(t), \sigma, w)$, $i=1,2$ with coordinate transformation $$\bm{\eta}^{\star}_i(v(t), \sigma, w)=Q_i \bm{\xi}_i(v(t), \sigma, w),\;\;\;\;i=1,2.$$
 Under Assumptions \ref{ass0}, \ref{ass0i} and \ref{ass5-explicit}, by Lemma 3 in \cite{wang2023nonparametric}, we have $$\bm{\eta}^{\star}_i(v(t), \sigma, w)=\theta_i,\;\;\;\;i=1,2.$$ 
Then, system \eqref{IM-01} yields the following dynamic compensator
\begin{subequations}\label{explicit-mas}
\begin{align}
\dot{\eta}_1&=M_1\eta_1+N_1x_2\label{explicit-mas1}\\
\dot{\eta}_2&=M_2\eta_2+N_2u.\label{explicit-mas2}
\end{align}
\end{subequations}
which is the internal model associated with systems \eqref{second-nonlinear-systems} and \eqref{eqn: exosystem system}.
\subsection{Error dynamics}
We now perform the following coordinate and input transformation on the composite system \eqref{second-nonlinear-systems} and \eqref{eqn: exosystem system}:
\begin{align*}
\bar{x}_{1}=&\,x_1-\mathbf{x}_1,\;\bar{\eta}_{1}=\eta_1- \theta_1-N_1\bar{x}_{1},\;
\bar{u}=u- \chi_2(\eta_2)\\
\bar{x}_{2}=&\,x_2- \mathbf{x}_2,\;\bar{\eta}_{2}=\eta_2- \theta_2-b(v, w)^{-1}N_2\bar{x}_{2}
\end{align*}
which yields an error system in the following form:
\begin{subequations}\label{Main-sys1}\begin{align}
\dot{\bar{\eta}}_{1}
=&\, M_1\bar{\eta}_1+M_1N_1\bar{x}_{1}\\
\dot{\bar{x}}_{1}
=&\,\bar{x}_2\\
 \dot{\bar{\eta}}_{2}
 =&\, M_2 \bar{\eta}_2+\bar{p}(\bar{x}_1, \bar{x}_2, \mu)\\
 \dot{\bar{x}}_{2}
=&\, \bar{f}(\bar{x}_1, \bar{x}_2, v, w)+  b(v,w)\tilde{\chi}_2(\bar{\eta}_2,\bar{x}_2,\mu) + b(v,w)\bar{u}
\end{align}\end{subequations}
where $ \mu=\col(\sigma,v,w)$,
\begin{align*}
\bar{f}(\bar{x}_1, \bar{x}_2, \mu)=&\,f(\bar{x}_1+\mathbf{x}_1, \bar{x}_2+\mathbf{x}_2, \mu)- f(\mathbf{x}_1, \mathbf{x}_2,\mu)\\
\bar{p}(\bar{x}_1, \bar{x}_2, \mu)=&\,-b(v, w)^{-1}N_2\bar{f}(\bar{x}_1, \bar{x}_2, \mu)\\
&+b(v, w)^{-1}M_2N_2\bar{x}_{2}-\frac{d b(v, w)^{-1}}{dt}N_2\bar{x}_{2}\\
\tilde{\chi}_2(\bar{\eta}_2,\bar{x}_2,\mu)=&\,\chi_2\left(
\bar{\eta}_{2}+ \bm{\eta}^{\star}_2+b(v, w)^{-1}N_2\bar{x}_{2}\right)-\chi_2(\bm{\eta}^{\star}_2). 
\end{align*}
%
%
%
It is noted that the signal $\bar{x}_2$ is not measurable since the signal $\mathbf{x}_2(v,w,\sigma)$ is not accessible, and that only $\chi_i(\eta_i)$ is available. Therefore, we define 
\begin{align}
\tilde{x}_2=&\, x_2-\chi_1(\eta_1),  
&\zeta= \tilde{x}_2+ \rho(\bar{x}_1)\bar{x}_1,
\end{align}
where $\rho(\cdot)$ is some smooth positive function to be specified later.
Then, we have $\tilde{x}_2=\bar{x}_2-\tilde{\chi}_1(\bar{\eta}_1, \bar{x}_{1})$ with $$\tilde{\chi}_1(\bar{\eta}_1, \bar{x}_{1})=\chi_1(\bar{\eta}_1+\bm{\eta}^{\star}_1+N_1\bar{x}_{1})-\chi_1(\bm{\eta}^{\star}_1),$$ which yields the following augmented system
\begin{subequations}\label{error-system}\begin{align}
\dot{\bar{\eta}}_{1}=&\,M_1\bar{\eta}_1+M_1N_1\bar{x}_{1}\label{error-system-a}\\
\dot{\bar{x}}_{1}=&\,\zeta-\rho(\bar{x}_1)\bar{x}_1+\tilde{\chi}_1(\bar{\eta}_1, \bar{x}_{1})\label{error-system-b}\\
\dot{\bar{\eta}}_{2}=&\,M_2 \bar{\eta}_2+\tilde{p}(\bar{x}_1, \bar{\eta}_1, \zeta, \mu)\label{error-system-c}\\
\dot{\zeta}=&\, \tilde{f}(\bar{x}_1, \bar{\eta}_1, \zeta, \mu)+b(v,w)\tilde{\chi}_2(\bar{x}_{1}, \bar{\eta}_1, \bar{\eta}_{2}, \zeta, \mu)+ b(v,w)\bar{u}\label{error-system-d}
\end{align}\end{subequations}
where 
\begin{align*}
\tilde{p}(\bar{x}_1, \bar{\eta}_1, \zeta, \mu)=&\, \bar{p}(\bar{x}_1, \zeta-\rho(\bar{x}_1)\bar{x}_1+\tilde{\chi}_1(\bar{\eta}_1, \bar{x}_{1}), \mu)\\
\tilde{f}(\bar{x}_1,  \bar{\eta}_1, \zeta, \mu)=&\,\bar{f}(\bar{x}_1, \zeta-\rho(\bar{x}_1)\bar{x}_1+\tilde{\chi}_1(\bar{\eta}_1, \bar{x}_{1}), \mu)\\
&-\frac{\partial \tilde{\chi}_1}{\partial \bar{\eta}_1} \left[M_1\bar{\eta}_1+M_1N_1\bar{x}_{1}\right]\\
&-\frac{\partial \tilde{\chi}_1}{\partial \bar{x}_1} \left[\zeta-\rho(\bar{x}_1)\bar{x}_1+\tilde{\chi}_1(\bar{\eta}_1, \bar{x}_{1})\right]\\
&+\frac{\partial (\rho(\bar{x}_1)\bar{x}_1)}{\partial \bar{x}_1}\big[\zeta-\rho(\bar{x}_1)\bar{x}_1+\tilde{\chi}_1(\bar{\eta}_1, \bar{x}_{1})\big]\\
\tilde{\chi}_2(\bar{x}_{1}, \bar{\eta}_1, \bar{\eta}_{2}, \zeta, \mu)=&\, \chi_2(\bar{\eta}_{2}+ \bm{\eta}^{\star}_2 +b(v, w)^{-1} N_2 (\zeta-\rho(\bar{x}_1)\bar{x}_1\\& +\tilde{\chi}_1(\bar{\eta}_1, \bar{x}_{1})))-\chi_2(\bm{\eta}^{\star}_2).
\end{align*}
It can be verified that all functions in  \eqref{error-system} are sufficiently smooth and vanish at the origin. 

\subsection{Nonadaptive robust output regulation}

For all $\mu \in \mathds{V}\times\mathds{W}\times\mathds{S}$, the origin is the equilibrium point of the augmented system \eqref{error-system} and the output $e$ is identically zero at the origin. By the general framework established in \cite{huang2004nonlinear}, the role of the internal model is to convert the output regulation problem of the system \eqref{second-nonlinear-systems} into the stabilization problem of the augmented system \eqref{error-system}, which is summarized as follows.

\begin{lem}\label{lemproconv1}
For all $\mu \in \mathds{V}\times\mathds{W}\times\mathds{S}$, if there exists a feedback control law of the form \begin{equation}\label{virtcontr}
\bar{u}=\psi(\zeta)
\end{equation}
that solves the robust stabilization problem of system \eqref{error-system}, then the robust output regulation problem of system \eqref{second-nonlinear-systems} can be solved by the control law composed of \eqref{explicit-mas} and \eqref{virtcontr}. 
\end{lem}

Then, we can introduce the following properties of the system \eqref{error-system} under Assumptions \ref{ass0}, \ref{ass0i} and \ref{ass5-explicit}.
\begin{lem}\label{lemmabodev} %
Consider system \eqref{error-system} subject to Assumptions \ref{ass0}, \ref{ass0i} and \ref{ass5-explicit}.  Then, for any smooth functions $\rho_1(\cdot)\geq 1$, $\gamma_{\bar{\eta}_1}(\cdot)\geq 1$ and $\gamma_{\bar{\eta}_2} (\cdot)\geq 1$, there exist a smooth positive function $\rho(\bar{x}_1)$ 
and a smooth input-to-state Lyapunov function $V_1:=V_1(\bar{\eta}_1, \bar{x}_1, \bar{\eta}_2)$ satisfying $$\underline{\alpha}_1(\|(\bar{\eta}_1, \bar{x}_1, \bar{\eta}_2)\|) \leq V_1(\bar{\eta}_1, \bar{x}_1, \bar{\eta}_2)\leq \bar{\alpha}_1(\|(\bar{\eta}_1, \bar{x}_1, \bar{\eta}_2)\|)$$ such that, for some comparison functions $\underline{\alpha}_1(\cdot)\in K_{\infty}$, $\bar{\alpha}_1(\cdot)\in K_{\infty}$ and positive smooth function $\hat{\beta}_{3p}(\cdot)$, 
\begin{align}\label{V1-derivative}
\dot{V}_1(\bar{\eta}_1, \bar{x}_1, \bar{\eta}_2) \leq & -\rho_1(\bar{x}_1)\bar{x}_1^2-\gamma_{\bar{\eta}_1}(\bar{\eta}_1 )\|\bar{\eta}_1\|^2\nonumber\\
&-\gamma_{\bar{\eta}_2} (\bar{\eta}_2 )\|\bar{\eta}_2\|^2+[\hat{\beta}_{3p}(\zeta)+2]\zeta^2
\end{align}
for all $\mu\in \mathds{V}\times\mathds{W}\times\mathds{S}$.
\end{lem}
\begin{proof}
Since $M_1$ and $M_2$ are Hurwitz, there exist positive definite matrices $P_1$ and $P_2$ such that $P_iM_i+M_iP_i^T\leq -4 I_i$, for $i=1,2$. In terms of equations \eqref{error-system-b} and \eqref{error-system-c}, 
it can be verified that the function $\bar{\chi}_1(\bar{\eta}_1, \bar{x}_{1})$ and $\tilde{p}(\bar{x}_1, \bar{\eta}_1, \zeta, \mu)$ are continuously differentiable and vanish at $\col(\bar{x}_{1}, \bar{\eta}_1, \zeta)=\col(0, 0, 0)$. 
Hence, by \cite[Lemma 11.1]{chen2015stabilization}, there exist positive smooth functions $\beta_{1\chi}(\cdot)$, $\beta_{2\chi}(\cdot)$, $\beta_{1p}(\cdot)$, and $\beta_{2p}(\cdot)$ such that for $\mu \in \mathds{V}\times\mathds{W}\times\mathds{S}$,
\begin{align*}
\|\tilde{\chi}_1(\bar{\eta}_1, \bar{x}_{1})\|^2\leq &\ \beta_{1\chi}(\bar{\eta}_1)\|\bar{\eta}_1\|^2+ \bar{x}_{1}^2\beta_{2\chi}(\bar{x}_{1}),\\
\|\tilde{p}(\bar{x}_1, \bar{\eta}_1, \zeta, \mu)\|^2\leq &\ \beta_{1 p}(\bar{\eta}_1)\|\bar{\eta}_1\|^2+ \bar{x}_{1}^2\beta_{2p}(\bar{x}_{1})+ \zeta^2\beta_{3p}(\zeta).
\end{align*}
Let $V_{\bar{\eta}_1}=\frac{1}{2} \bar{\eta}_1^TP_1\bar{\eta}_1$ and $V_{\bar{\eta}_2}=\frac{1}{2} \bar{\eta}_2^TP_2\bar{\eta}_2$. Then, the time derivative of $V_{\bar{\eta}_1}$ and $V_{\bar{\eta}_2}$ along the trajectories of systems \eqref{error-system-b} and \eqref{error-system-c} are as follows{\myr 
\begin{align*}
\dot{V}_{\bar{\eta}_1}\leq&\ (1/2)\bar{\eta}_1^T\left(P_1M_1+M_1P_1^T\right)\bar{\eta}_1+\bar{\eta}_1^TP_1M_1N_1\bar{x}_{1}\\
\leq&- \|\bar{\eta}_1\|^2+\frac{1}{4}\|P_1M_1N_1\|^2\|\bar{x}_{1}\|^2,\\
\dot{V}_{\eta_2}
\leq&\ (1/2)\bar{\eta}_2^T\left(P_2M_2+M_2P_2^T\right)\bar{\eta}_2+\bar{\eta}_2^T\bar{p}(\bar{x}_1, \bar{x}_2, \mu)\\
\leq& - \|\bar{\eta}_2\|^2+\|P_2\|^2\big[\beta_{1 p}(\bar{\eta}_1)\|\bar{\eta}_1\|^2\\
&+ \bar{x}_{1}^2\beta_{2p}(\bar{x}_{1})+ \zeta^2\beta_{3p}(\zeta)\big].
\end{align*}}
By using the changing supply rate technique \cite{sontag1995changing}, given any smooth functions $\bar{\gamma}_{\bar{\eta}_1} (\bar{\eta}_1 )>0$ and $\gamma_{\bar{\eta}_2} (\bar{\eta}_2 )>0$, there exists a continuous function $U_{\bar{\eta}_i}:=U_{\bar{\eta}_i}(\bar{\eta}_i)$ satisfying 
$$\underline{\alpha}_{\bar{\eta}_i}\big(\big\|\bar{\eta}_i \big\|\big)\leq U_{\bar{\eta}_i}(\bar{\eta}_i)\leq\overline{\alpha}_{\bar{\eta}_i}\big(\big\|\bar{\eta}_i \big\|\big),\;i=1,2$$ for some class $\mathcal{K}_{\infty}$ functions $\underline{\alpha}_{\bar{\eta}_i}(\cdot)$ and $\overline{\alpha}_{\bar{\eta}_i}(\cdot)$, such that for all $\mu \in \mathds{V}\times\mathds{W}\times\mathds{S}$
\begin{align*}\dot{U}_{\bar{\eta}_1}\leq &-\bar{\gamma}_{\bar{\eta}_1} (\bar{\eta}_1 )\|\bar{\eta}_1\|^2+\hat{\beta}_{u1}(\bar{x}_1)\|\bar{x}_1\|^2,\\
\dot{U}_{\bar{\eta}_2}\leq&-\gamma_{\bar{\eta}_2} (\bar{\eta}_2 )\|\bar{\eta}_2\|^2+\hat{\beta}_{1 p}(\bar{\eta}_1)\|\bar{\eta}_1\|^2\\
&+ \hat{\beta}_{2p}(\bar{x}_{1})\bar{x}_{1}^2+ \hat{\beta}_{3p}(\zeta)\zeta^2,
\end{align*}
for some positive smooth functions $$\hat{\beta}_{u1}(\cdot)\geq 1,~~\hat{\beta}_{1p}(\cdot)\geq 1,~~\hat{\beta}_{2p}(\cdot)\geq 1~~\textnormal{and}~~ \hat{\beta}_{3p}(\cdot)\geq 1.$$ 
Let $V_{\bar{x}_1}(\bar{x}_1)=\frac{1}{2}\bar{x}_1^2$ and the time derivative along the trajectories of \eqref{error-system-b} satisfies the following inequality
\begin{align*}
\dot{V}_{\bar{x}_1}\leq &\ \bar{x}_1 \zeta-\rho(\bar{x}_1)\bar{x}_1^2+\bar{x}_1\tilde{\chi}_1(\bar{\eta}_1, \bar{x}_{1})\nonumber\\
&- \|\bar{\eta}_2\|^2+\|P_2\|^2\|\tilde{p}(\bar{x}_1, \bar{\eta}_1, \zeta, \mu)\|^2\\
\leq &\  2\zeta^2-\rho(\bar{x}_1)\bar{x}_1^2+2\beta_{1\chi}(\bar{\eta}_1)\|\bar{\eta}_1\|^2+ \bar{x}_{1}^2(2\beta_{2\chi}(\bar{x}_{1})+4)\\
&- \|\bar{\eta}_2\|^2+\|P_2\|^2[\beta_{1 p}(\bar{\eta}_1)\|\bar{\eta}_1\|^2+ \bar{x}_{1}^2\beta_{2p}(\bar{x}_{1})+ \zeta^2\beta_{3p}(\zeta)]\\
\leq&-\big[\rho(\bar{x}_1)-2\beta_{2\chi}(\bar{x}_{1})-1\big]\bar{x}_1^2 +2\beta_{1\chi}(\bar{\eta}_1)\|\bar{\eta}_1\|^2+2\zeta^2.
\end{align*}
Let $V_1(\bar{\eta}_1, \bar{x}_{1}, \bar{\eta}_2)=V_0(\bar{x}_1)+U_{\bar{\eta}_1}(\bar{\eta}_1)+U_{\bar{\eta}_2}(\bar{\eta}_2)$. Then, for all $\mu \in \mathds{V}\times\mathds{W}\times\mathds{S}$, the time derivative along the trajectory of $(\bar{\eta}_1, \bar{x}_{1}, \bar{\eta}_2)$-subsystem in \eqref{error-system} results in
\begin{align*}
\dot{V}_{1} \leq &-\big[\rho(\bar{x}_1)-2\beta_{2\chi}(\bar{x}_{1})-1- \hat{\beta}_{2p}(\bar{x}_{1})-\hat{\beta}_{u1}(\bar{x}_1)\big]\bar{x}_1^2\\
&-\big[\bar{\gamma}_{\bar{\eta}_1} (\bar{\eta}_1 )-2\beta_{1\chi}(\bar{\eta}_1)- \hat{\beta}_{1 p}(\bar{\eta}_1)\big]\|\bar{\eta}_1\|^2\\
&- \gamma_{\bar{\eta}_2} (\bar{\eta}_2 )\|\bar{\eta}_2\|^2+\big[ \hat{\beta}_{3p}(\zeta)+2\big]\zeta^2.
\end{align*}
Finally, inequality \eqref{V1-derivative} can be achieved by letting $$\rho(\bar{x}_1)\geq 2\beta_{2\chi}(\bar{x}_{1})+1+ \hat{\beta}_{2p}(\bar{x}_{1})+\hat{\beta}_{u1}(\bar{x}_1)$$ and $$\bar{\gamma}_{\bar{\eta}_1} (\bar{\eta}_1 )\geq 2\beta_{1\chi}(\bar{\eta}_1)+\hat{\beta}_{1 p}(\bar{\eta}_1)+\gamma_{\bar{\eta}_1}(\bar{\eta}_1 ).$$
\end{proof}

\begin{lem}\label{lem-stab}%
Under Assumptions \ref{ass0}, \ref{ass0i} and \ref{ass5-explicit}, 
there exist a 
positive smooth function $k(\cdot)$ and a positive number $k^*$ such that the following nonadatpive control law  
\begin{align}\label{ESC-1b}
\bar{u} &= -k_0k(\zeta)\zeta
\end{align}
where $k_0\geq k^*$ and $\zeta= x_2-\chi_1(\eta_1)+ \rho(\bar{x}_1)\bar{x}_1$,
solves the robust stabilization problem for system \eqref{error-system}.
\end{lem}%

\begin{proof} Regarding \eqref{error-system-d}, it can
be verified that the functions $\tilde{f}(\bar{x}_1, \bar{\eta}_1, \zeta, \mu)$ and \begin{align*}\tilde{\chi}_2(\bar{x}_{1}, \bar{\eta}_1, \bar{\eta}_{2}, \zeta, \mu)
&=-\chi_2(\bm{\eta}^{\star}_2)+\chi_2\Big(\bar{\eta}_{2}+\bm{\eta}^{\star}_2\\
& +b(\mu)^{-1} N_2 \big(\zeta-\rho(\bar{x}_1)\bar{x}_1+\tilde{\chi}_1(\bar{\eta}_1, \bar{x}_{1})\big)\Big) \end{align*} 
vanish at $\col(\bar{x}_{1}, \bar{\eta}_1, \bar{\eta}_2, \zeta)=\col(0, 0, 0, 0)$ for all $\mu \in \mathds{V}\times\mathds{W}\times\mathds{S}$. Hence, by \cite[Lemma 11.1]{chen2015stabilization}, there exist positive smooth functions $\beta_{1f}(\cdot)$, $\beta_{2f}(\cdot)$, $\beta_{3f}(\cdot)$, $\beta_{3\chi}(\cdot)$, $\beta_{4\chi}(\cdot)$,  $\beta_{5\chi}(\cdot)$ and $\beta_{6\chi}(\cdot)$ such that, for all $\mu \in \mathds{V}\times\mathds{W}\times\mathds{S}$,
$$\|\tilde{f}(\bar{x}_1, \bar{\eta}_1, \zeta, \mu)\|^2\leq  \beta_{1f}(\bar{x}_{1})\bar{x}_{1}^2+\beta_{2f}(\bar{\eta}_1)\|\bar{\eta}_1\|^2+\beta_{3f}( \zeta) \zeta^2$$ and
\begin{align*}\|b(v,w)\tilde{\chi}_2(\bar{x}_{1}, \bar{\eta}_1, \bar{\eta}_{2}, \zeta, \mu)\|^2\leq&\ \beta_{3\chi}(\bar{x}_1)\bar{x}_1^2+\beta_{4\chi}(\bar{\eta}_1)\|\bar{\eta}_1\|^2\\
&+\beta_{5\chi}(\bar{\eta}_{2})\|\bar{\eta}_{2}\|^2+\beta_{6\chi}(\zeta) \zeta^2.
\end{align*}
We pose a Lyapunov function candidate, $$V:=V(\bar{x}_1, \bar{\eta}_1, \bar{\eta}_2, \zeta),$$ for the closed-loop system \eqref{error-system} and \eqref{ESC-1b}, given by
\begin{align*}
V(\bar{\eta}_1, \bar{x}_1, \bar{\eta}_2, \zeta)=V_1(\bar{\eta}_1, \bar{x}_1, \bar{\eta}_2)+ \zeta^2.
\end{align*}
By Lemma \ref{lemmabodev}, for any $\mu \in \mathds{V}\times\mathds{W}\times\mathds{S}$, along the trajectory of \eqref{error-system} and \eqref{ESC-1b}, we have 
\begin{equation}
\begin{split}
    \dot{V}\leq & -\rho_1(\bar{x}_1)\bar{x}_1^2-\gamma_{\bar{\eta}_1}(\bar{\eta}_1 )\|\bar{\eta}_1\|^2-\gamma_{\bar{\eta}_2} (\bar{\eta}_2 )\|\bar{\eta}_2\|^2\\
&+\big[\hat{\beta}_{3p}(\zeta)+3+\|b(v,w)\|^2\big]\zeta^2+\|\tilde{f}(\bar{x}_1, \bar{\eta}_1, \zeta, \mu)\|^2\\
& +\|b(v,w)\tilde{\chi}_2(\bar{x}_{1}, \bar{\eta}_1, \bar{\eta}_{2}, \zeta, \mu)\|^2- 2b(v, w) k_0 k(\zeta) \zeta^2\\
\leq & -\big[\rho_1(\bar{x}_1)-\beta_{1f}(\bar{x}_{1})-\beta_{3\chi}(\bar{x}_1)\big]\bar{x}_1^2\\
&-\big[\gamma_{\bar{\eta}_1}(\bar{\eta}_1 )-\beta_{2f}(\bar{\eta}_1)-\beta_{4\chi}(\bar{\eta}_1)\big]\|\bar{\eta}_1\|^2\\
&- \big[\gamma_{\bar{\eta}_2} (\bar{\eta}_2 )-\beta_{5\chi}(\bar{\eta}_{2})\big]\|\bar{\eta}_2\|^2+\big[\hat{\beta}_{3p}(\zeta)+3\\
&+\|b(v,w)\|^2+\beta_{3f}( \zeta)+\beta_{6\chi}(\zeta)- 2b(v, w) k_0 k(\zeta) \big]\zeta^2.\nonumber
\end{split}
\end{equation}
It is noted that 
\begin{equation}\label{theoeq1}
\hat{\beta}_{3p}(\zeta)+3+\|b(v,w)\|^2+\beta_{3f}( \zeta)+\beta_{6\chi}(\zeta)\leq  \beta^* \beta_{\max}(\zeta)
\end{equation}
where $\beta^*$ is some positive constant and $\beta_{\max}(\zeta)$ is some known positive smooth function.
Let the smooth functions be chosen as follows
\begin{subequations}
\begin{align}
\rho_1(\bar{x}_1) \geq &\, \beta_{1f}(\bar{x}_{1})+\beta_{3\chi}(\bar{x}_1)+1 \label{Theofun1-1}\\
\gamma_{\bar{\eta}_1}(\bar{\eta}_1 )\geq &\, \beta_{2f}(\bar{\eta}_1)+\beta_{4\chi}(\bar{\eta}_1)+1 \label{Theofun1-2}\\
\gamma_{\bar{\eta}_2} (\bar{\eta}_2 )\geq &\, \beta_{5\chi}(\bar{\eta}_{2})+1 \label{Theofun1-3}\\
k(\zeta)\geq &\, \beta_{\max}(\zeta)+1 \label{Theofun1-4}\\
k_0\geq &\, k^*:=\max\{{\beta^*}/{2b^*},1\}\label{Theofun1-5}
\end{align}\end{subequations}
where $b^*$ is the lower bound of $b(v, w)$.\\
Hence,  we have 
\begin{align}\label{dotV}\dot{V}(\bar{\eta}_1, \bar{\eta}_2, \bar{x}_1, \zeta)\leq -\big\|\bar{\eta}_1\big\|^2-\big\|\bar{\eta}_2\big\|^2-\bar{x}_1^2-\zeta^2.
\end{align}
That is, for all $\mu \in \mathds{V}\times\mathds{W}\times\mathds{S}$, the equilibrium of the closed-loop system at the origin is
globally asymptotically stable. The proof is thus completed.
\end{proof}

Based on Lemma \ref{lem-stab}, we obtain the solution to Problem \ref{Prob: second-order-Output-regulation} summarized as follows.

\begin{thm}\label{Theorem-1}%
Under Assumptions \ref{ass0}, \ref{ass0i} and \ref{ass5-explicit}, the output regulation problem of the nonlinear uncertain system \eqref{second-nonlinear-systems} can be solved by the nonadaptive control law 
\begin{subequations}\begin{align}\label{}
\dot{\eta}_1&=M_1\eta_1+N_1x_2\\
\dot{\eta}_2&=M_2\eta_2+N_2u\\
u &= -k_0k(\zeta)\zeta + \chi_2(\eta_2)
\end{align}\end{subequations}
where $M_i$ and $N_i$ are defined in \eqref{MNINter}, $k(\cdot)$ is defined in \eqref{Theofun1-4}, $k_0$ is defined in \eqref{Theofun1-5}, and $\zeta= x_2-\chi_1(\eta_1)+ \rho(e)e$.
\end{thm}%
{\myr By using Lemma \ref{lem-stab}, Theorem \ref{Theorem-1} can be directly obtained from Lemma 3. Therefore, to avoid redundancy, the detailed proof is omitted from the manuscript. We have clarified this point for better understanding.}
\begin{Corollary}\label{Theorem-2}%
Under Assumptions \ref{ass0}, \ref{ass0i} and \ref{ass5-explicit}, the output regulation problem of the nonlinear uncertain system \eqref{second-nonlinear-systems} can be solved by the control law 
\begin{subequations}\label{ctr2}
\begin{align}
\dot{\eta}_1&=M_1\eta_1+N_1x_2\\
\dot{\eta}_2&=M_2\eta_2+N_2u\\
\dot{\hat{k}}&=k(\zeta)\zeta^2\\
u &= -\hat{k} k(\zeta)\zeta + \chi_2(\eta_2)\label{adESC-1b}
\end{align}
\end{subequations}
where $M_i$ and $N_i$ are defined in \eqref{MNINter} with $i=1, 2$, $k(\cdot)$ is defined in \eqref{Theofun1-4}, $k_0$ is defined in \eqref{Theofun1-5}, and $\zeta= x_2-\chi_1(\eta_1)+ \rho(e)e$.
\end{Corollary}%
\begin{proof}
Define the following Lyapunov function  
$$W(\bar{\eta}_1, \bar{\eta}_2, \bar{x}_1, \zeta)=V(\bar{\eta}_1, \bar{\eta}_2, \bar{x}_1, \zeta)+b^*(\hat{k}-\bar{k})^2,$$ where $\bar{k}$ is some positive number to be determined. Then, for any $\mu \in \mathds{V}\times\mathds{W}\times\mathds{S}$, along the trajectory of \eqref{error-system} and \eqref{ctr2}, the time derivative of $W(\bar{\eta}_1, \bar{\eta}_2, \bar{x}_1, \zeta)$ admits
\begin{equation}
\begin{split}
    \dot{W}\leq & -\rho_1(\bar{x}_1)\bar{x}_1^2-\gamma_{\bar{\eta}_1}(\bar{\eta}_1 )\|\bar{\eta}_1\|^2-\gamma_{\bar{\eta}_2} (\bar{\eta}_2 )\|\bar{\eta}_2\|^2\\
&+\big[\hat{\beta}_{3p}(\zeta)+3+\|b(v,w)\|^2\big]\zeta^2+\|\tilde{f}(\bar{x}_1, \bar{\eta}_1, \zeta, \mu)\|^2\\
& +\|b(v,w)\tilde{\chi}_2(\bar{x}_{1}, \bar{\eta}_1, \bar{\eta}_{2}, \zeta, \mu)\|^2- 2b(v, w) \hat{k} k(\zeta) \zeta^2\\&+2b^*(\hat{k}-\bar{k})\dot{\hat{k}}\\
\leq & -\rho_1(\bar{x}_1)\bar{x}_1^2-\gamma_{\bar{\eta}_1}(\bar{\eta}_1 )\|\bar{\eta}_1\|^2-\gamma_{\bar{\eta}_2} (\bar{\eta}_2 )\|\bar{\eta}_2\|^2\\
&+\big[\hat{\beta}_{3p}(\zeta)+3+\|b(v,w)\|^2\big]\zeta^2+\|\tilde{f}(\bar{x}_1, \bar{\eta}_1, \zeta, \mu)\|^2\\
& +\|b(v,w)\tilde{\chi}_2(\bar{x}_{1}, \bar{\eta}_1, \bar{\eta}_{2}, \zeta, \mu)\|^2- 2b^* \hat{k} k(\zeta) \zeta^2\\&+2b^*(\hat{k}-\bar{k})k(\zeta)\zeta^2\\
\leq & -\big[\rho_1(\bar{x}_1)-\beta_{1f}(\bar{x}_{1})-\beta_{3\chi}(\bar{x}_1)\big]\bar{x}_1^2\\
&-\big[\gamma_{\bar{\eta}_1}(\bar{\eta}_1 )-\beta_{2f}(\bar{\eta}_1)-\beta_{4\chi}(\bar{\eta}_1)\big]\|\bar{\eta}_1\|^2\\
&- \big[\gamma_{\bar{\eta}_2} (\bar{\eta}_2 )-\beta_{5\chi}(\bar{\eta}_{2})\big]\|\bar{\eta}_2\|^2+\big[\hat{\beta}_{3p}(\zeta)+3\\
&+\|b(v,w)\|^2+\beta_{3f}( \zeta)+\beta_{6\chi}(\zeta)- 2b^* \bar{k} k(\zeta) \big]\zeta^2.\nonumber
\end{split}
\end{equation}
%
In view of \eqref{theoeq1},
letting $\rho_1(\cdot)$, $\gamma_{\bar{\eta}_1}(\cdot)$, $\gamma_{\bar{\eta}_2}(\cdot)$  and $k(\cdot)$ be the smooth functions defined in \eqref{Theofun1-1}-\eqref{Theofun1-4}, and $\bar{k}\geq \max\{{\beta^*}/({2b^*}),1\}$ gives 
\begin{align} \dot{W}(\bar{\eta}_1, \bar{\eta}_2, \bar{x}_1, \zeta)\leq -\big\|\bar{\eta}_1\big\|^2-\big\|\bar{\eta}_2\big\|^2-\bar{x}_1^2-\zeta^2
\end{align}
 for all $\mu \in \mathds{V}\times\mathds{W}\times\mathds{S}$. The proof is thus completed by invoking Lemma \ref{lem-stab}.
\end{proof}

\begin{rem}
 From Lemma 3 in \cite{wang2023nonparametric}, the existence of the nonlinear mappings $\chi_i$
 in \eqref{IM-01} relies on the solution of a time-varying equation 
$$\Theta(\eta_i)\check{a}_i (\eta_i)+\textnormal{\col}(\eta_{i,n_i +1},\cdots,\eta_{i,2n_i})=0,\;\;i=1,2.$$
It is noted that the inverse of 
$\Theta(\eta_i(t))$ may not be well-defined
for all $t\geq 0$.
Therefore, for the matrix $\Theta \in \mathds{R}^{n\times n}$,  we can define the following globally defined and smooth mapping 
$$\check{a}_{i}\left(\eta_i\right)=-O(\Theta(\eta_i))\textnormal{\col}(\eta_{i,n_i +1},\cdots,\eta_{i,2n_i})$$
where $O(\Theta):\mathbb{R}^{n\times n}\rightarrow  \mathds{R}^{n\times n}$ is a  function defined as follows:
\begin{equation}\label{eq-Theta}
 \begin{split}
      O(\Theta)
      =\frac{\textnormal{\mbox{det}}[\Theta]}{\textnormal{\mbox{det}}^2[\Theta]+\Psi(1+\textnormal{\mbox{det}}^2[\Theta]-\epsilon^2)}\textnormal{\mbox{Adj}}[\Theta]
 \end{split}
\end{equation}
\begin{equation*} 
  \begin{split}
  \Psi(\varsigma)=\frac{\kappa(1-s)}{\kappa(s)+\kappa(1-s)},\;   \kappa(s)=\left\{ {\begin{array}{*{20}{c}}
{e^{-1/s},}&{\mbox{if}\ s > 0}\\
0, &{\mbox{else}}
\end{array}} \right.
  \end{split}
\end{equation*}
 for a positive constant $\epsilon$.
 Then, from \cite{wang2023nonparametric},
 the smooth function $\chi_i(\eta_i)$ is given by:
 \begin{align*}\myr 
\chi_i(\eta_i)=\Gamma_i \Xi_i (\check{a}_{i} )\textnormal{\col}(\eta_{i,1},\cdots,\eta_{i,n_i}), \;\;i=1,2.
\end{align*}
\end{rem}
 
\begin{rem}
In comparison to the adaptive internal model approach, the nonadaptive approach is developed to solve the output regulation problem of nonlinear uncertain systems subject to unknown exosystems. 
The nonadaptive internal model brings some features. First, as shown in Lemma \ref{lemmabodev}, the inverse dynamics associated with the
internal model in the augmented system \eqref{error-system} is input-to-state stable, which is beneficial for the nonadaptive stabilization control design. 
Second, the proposed robust control approach does not need the additional adaptive dynamics for compensation of unknown parameters caused by the unknown exosystem \cite{su2013cooperative,liu2009parameter}. 
Third, a strict Lyapunov function is constructed, and the system stability is established in the sense of Lyapunov. In comparison, it is generally required to apply Barbalat's lemma for the stability analysis of the adaptive internal model approach. Thus, improved robustness with respect to unmodelled disturbances can be guaranteed for the proposed generic internal model approach in contrast to the adaptive internal model approach \cite{su2013cooperative}.
\end{rem}
\section{Application to the Duffing system control}\label{numerexam}
Consider the nonlinear system modelled by a controlled Duffing system as follows:
\begin{subequations}\label{ex:duffin}
    \begin{align}
    \dot{x}_1=&\ x_2  \\
    \dot{x}_2=& -c_3 x_2-c_1x_1-c_2x_1^3+u+d(t)\\
    y=&\ x_1
    \end{align}
\end{subequations}
where $\col(x_1,x_2)\in \mathds{R}^2$ is the state, 
$c_1$, $c_2$, and $c_3$ are unknown system parameters, and $c=\mbox{col}(c_1,c_2,c_3) \in \mathds{W}=\{c \in \mathds{R}^3: c_i \in [-2,2], i=1,2,3\}$.
The unknown external disturbance $d(t)$
is 
generated by the uncertain exosystem of the form:
\begin{subequations}\label{ex:exosystem} \begin{align}
 \dot{v}
 =&\,\left[\begin{matrix}0 &\sigma \\ -\sigma & 0\end{matrix}\right]v\\
 d =&\, \left[\begin{matrix}0 &1 \end{matrix}\right]v\\
 e =&\, y - v_1
 \end{align}\end{subequations}
where $\sigma\in \mathds{S}=\{\sigma\in \mathds{R}: \sigma \in [0.1,2]\}$ and $ \mathds{V}=\{v(0)\in \mathds{R}^2: v_i(0)  \in [-3,3] \}$.

The solution of the regulator equation associated with \eqref{ex:duffin} and \eqref{ex:exosystem}
is $\mathbf{x}_1(v,w,\sigma)= v_1$, $\mathbf{x}_2(v,w,\sigma)=\sigma v_2$, and
    $\mathbf{u}_2(v,w,\sigma)= -\sigma^2v_1+c_3 \sigma v_2+c_1 v_1+ c_2 v_1^3 - v_2 $.
Then, we have
 \begin{align*}\Phi_1 (a_1(\sigma))=&\,\left[\begin{matrix}0 & 1\\
-\sigma^2 & 0\end{matrix}\right],\; &\Gamma_1=\left[\begin{matrix}1 & 0\end{matrix}\right],\\
\Phi_2 (a_2(\sigma))=&\,\left[\begin{matrix}0_{3\times 1}& I_3\\
-9\sigma^2, & 0,-10\sigma^2,0\\
\end{matrix}\right],\; &\Gamma_2=\left[\begin{matrix}1 & 0_{1\times 3}\end{matrix}\right],
 \end{align*}
where $a_1(\sigma)=\col(\sigma^2,0)$ and $a_2(\sigma)=\col(9\sigma^4,0,10\sigma^2,0)$.
It can be verified that Assumptions \ref{ass0}--\ref{ass5-explicit} are all satisfied. Then, we let $\epsilon=0.1$, $m_1=\col(10, 18, 15, 6)$ and $m_2=\col(1, 5, 13, 22, 26, 22, 13, 5)$. 
 Then, we can obtain $\chi_1(\eta_1, \check{a}_{1})$,  $ \chi_2(\eta_2, \check{a}_{2})$ as in Table I, and choose $\rho(s)=10+4s^4$, $k_0=1$, $k(s)=s^2+1 $. 
 
The simulation is conducted with the unknown system parameters $c_1=-2$, $c_2=1.5$, $c_3=0.5$, $\sigma=0.5$, and the random chosen initial conditions with $x(0)=\col(1,-1)$, $v(0)=\col(1,1)$, $\eta_1(0)=0_4$, and $\eta_2(0)=0_8$.
The system trajectory and the tracking error are shown in Figs. \ref{figtraj} and \ref{figerror}, respectively. 
The tracking error is shown to vanish, and that, correspondingly, the state variables track the required steady-state trajectory. Moreover, the non-adaptive parameter estimation errors are shown to vanish as shown in Fig. \ref{fig5}. The simulation study thus demonstrates that the robust output regulation problem is effectively solved by the proposed approach for this uncertain system. 

   
\begin{table*}[ht]\myr 
\caption{Explicit Solution of the nonlinear mapping}\label{table_nonlinear}
\hrule \hrule\hrule\hrule   
{\begin{align*}
 \chi_1(\eta_1, \check{a}_{1})=&\,\eta_{1,1} (\check{a}_{1,1}^2 -{  m_{1,3}} \check{a}_{1,1} + { m_{1,1}}) + \eta_{1,2} ({ m_{1,2}} - \check{a}_{1,1} { m_{1,4}}),&\\
 \chi_2(\eta_2, \check{a}_{2})=&\,\eta_3({ m_{2,3} }+ \check{a}_{2,1}\check{a}_{2,3} - \check{a}_{2,3}{ m_{2,5}} + \check{a}_{2,3}(\check{a}_{2,1}- \check{a}_{2,3}^2 ) - { m_{2,7}}(\check{a}_{2,1}- \check{a}_{2,3}^2))+ \eta_2({ m_{2,2}} - \check{a}_{2,1}{ m_{2,6}} + \check{a}_{2,1}\check{a}_{2,3}{ m_{2,8}})\\
 & + \eta_1(\check{a}_{2,1}^2 - \check{a}_{2,1}\check{a}_{2,3}^2 + { m_{2,7}}\check{a}_{2,1}\check{a}_{2,3} - {m_{2,5}}\check{a}_{2,1} + {m_{2,1}})- \eta_4(\check{a}_{2,3}{m_{2,6}} -{ m_{2,4}} + { m_{2,8}}(\check{a}_{2,1}- \check{a}_{2,3}^2)),&
  \end{align*}}
\hrule \hrule 
{  \begin{align*}
 \check{a}_{1,1}=&\,\frac{\textnormal{\mbox{det}}[\Theta(\eta_1)]}{\textnormal{\mbox{det}}^2[\Theta(\eta_1)]+\Psi(1+\textnormal{\mbox{det}}^2[\Theta(\eta_1)]-\epsilon^2)}\left(\eta_{1,2}\eta_{1,4}-  \eta_{1,3}^2\right),&\check{a}_{1,2}=0,\\
 \check{a}_{2,1}=&\,\frac{\textnormal{\mbox{det}}[\Theta(\eta_2)]}{\textnormal{\mbox{det}}^2[\Theta(\eta_2)]+\Psi(1+\textnormal{\mbox{det}}^2[\Theta(\eta_2)]-\epsilon^2)} \big[\eta_{2,5}(\eta_{2,7}\eta_{2,4}^2 - 2 \eta_{2,4} \eta_{2,5} \eta_{2,6} + \eta_{2,5}^3 - \eta_{2,3} \eta_{2,7} \eta_{2,5} + \eta_{2,3} \eta_{2,6}^2) \\
 &- \eta_{2,8} (\eta_{2,6} \eta_{2,3}^2 - 2 \eta_{2,3} \eta_{2,4} \eta_{2,5} + \eta_{2,4}^3 - \eta_{2,2} \eta_{2,6} \eta_{2,4} + \eta_{2,2} \eta_{2,5}^2) \\
 &- \eta_{2,7} (- \eta_{2,7} \eta_{2,3}^2 + \eta_{2,6} \eta_{2,3} \eta_{2,4} + \eta_{2,3} \eta_{2,5}^2 - \eta_{2,4}^2 \eta_{2,5} + \eta_{2,2} \eta_{2,7} \eta_{2,4} - \eta_{2,2} \eta_{2,6} \eta_{2,5}) \\
 &- \eta_{2,6} (- \eta_{2,4}^2 \eta_{2,6} + \eta_{2,4} \eta_{2,5}^2 + \eta_{2,3} \eta_{2,7} \eta_{2,4} - \eta_{2,3} \eta_{2,5} \eta_{2,6} - \eta_{2,2} \eta_{2,7} \eta_{2,5} + \eta_{2,2} \eta_{2,6}^2)\big], &\check{a}_{2,2}=0,\\
 \check{a}_{2,3}=&\,\frac{\textnormal{\mbox{det}}[\Theta(\eta_2)]}{\textnormal{\mbox{det}}^2[\Theta(\eta_2)]+\Psi(1+\textnormal{\mbox{det}}^2[\Theta(\eta_2)]-\epsilon^2)} \big[\eta_{2,8} (- \eta_{2,6} \eta_{2,2}^2 + \eta_{2,5} \eta_{2,2} \eta_{2,3} + \eta_{2,2} \eta_{2,4}^2 - \eta_{2,3}^2 \eta_{2,4} + \eta_{2,1} \eta_{2,6} \eta_{2,3} - \eta_{2,1} \eta_{2,5} \eta_{2,4}) \\
&- \eta_{2,5} (- \eta_{2,7} \eta_{2,3}^2 + \eta_{2,6} \eta_{2,3} \eta_{2,4} + \eta_{2,3} \eta_{2,5}^2 - \eta_{2,4}^2 \eta_{2,5} + \eta_{2,2} \eta_{2,7} \eta_{2,4} - \eta_{2,2} \eta_{2,6} \eta_{2,5}) \\
&- \eta_{2,6} (\eta_{2,4}^3 - \eta_{2,1} \eta_{2,4} \eta_{2,7} + \eta_{2,1} \eta_{2,5} \eta_{2,6} + \eta_{2,2} \eta_{2,3} \eta_{2,7} - \eta_{2,2} \eta_{2,4} \eta_{2,6} - \eta_{2,3} \eta_{2,4} \eta_{2,5}) \\
&+ \eta_{2,7} (\eta_{2,7} \eta_{2,2}^2 - 2 \eta_{2,2} \eta_{2,4} \eta_{2,5} + \eta_{2,3} \eta_{2,4}^2 + \eta_{2,1} \eta_{2,5}^2 - \eta_{2,1} \eta_{2,3} \eta_{2,7})\big],&\check{a}_{2,4}=0. 
 \end{align*} }
\hrule \hrule \hrule\hrule  
\end{table*}

\begin{figure}[htbp]
\centering
\epsfig{figure=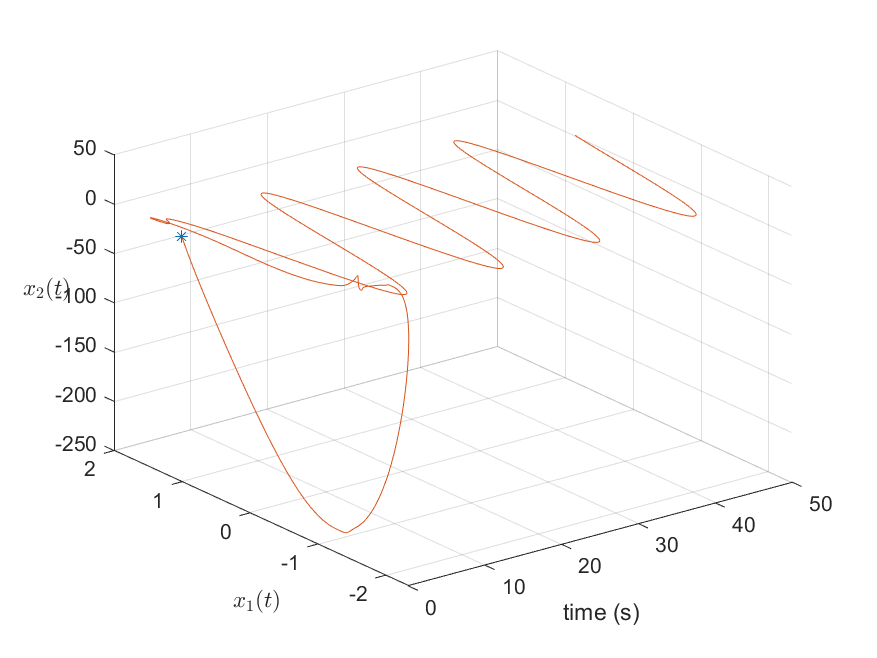,width=0.5\textwidth}
\caption{Trajectory of $(t,x_1,x_2)$ of the Duffing system}\label{figtraj}
\quad
\epsfig{figure=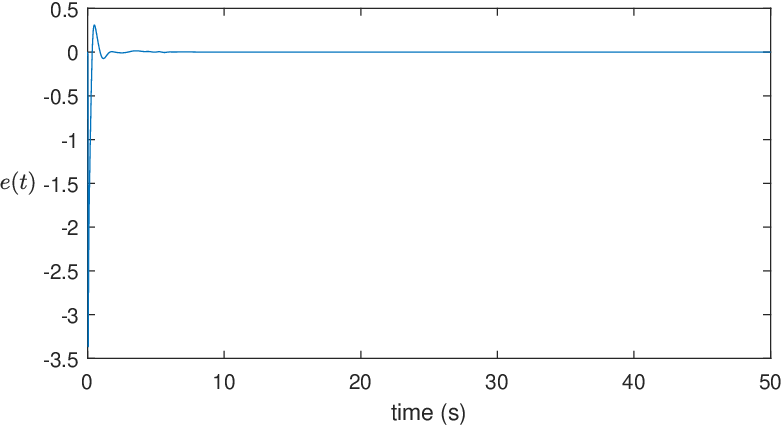,width=0.5\textwidth}
\caption{\myr Tracking error of the Duffing system}\label{figerror}
\quad
 \epsfig{figure=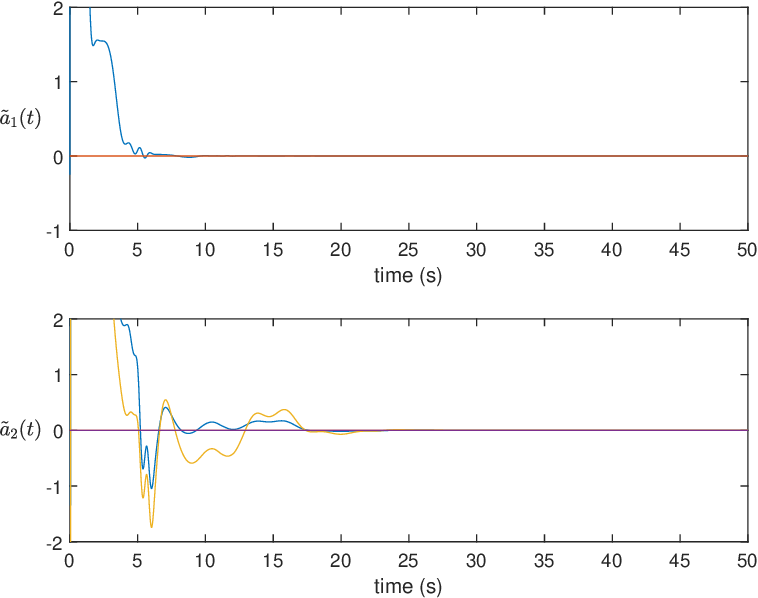,width=0.5\textwidth}
 \caption{\myr Parameter estimation error of non-adaptive method}\label{fig5}
 \end{figure}

\section{Conclusion}\label{conlu}
 
The robust output regulation problem has been addressed for a class of second-order nonlinear uncertain systems subject to an unknown exosystem. By analyzing the steady state of the system,  generic internal models are constructed that achieve exact steady-state compensation. A coordinate transformation is introduced that converts the robust output regulation problem to a stabilization problem of an augmented system composed of the second-order nonlinear uncertain system and the generic internal models. For the proposed robust stabilization control law, a strict Lyapunov function is constructed. With the help of the changing supply rate technique, it is shown that the output zeroing manifold of the augmented system can be made attractive by the proposed nonadaptive control law, which is shown to solve the robust output regulation problem. 
Finally, the effectiveness of the proposed nonadaptive internal model approach is demonstrated by its application to the control of a Duffing system.

\noindent
\footnotesize
\bibliographystyle{ieeetr}
\bibliography{myref}

\begin{thebibliography}{10}

\bibitem{isidori1990output}
A.~Isidori and C.~I. Byrnes, ``Output regulation of nonlinear systems,'' {\em IEEE Transactions on Automatic Control}, vol.~35, no.~2, pp.~131--140, 1990.

\bibitem{huang2004nonlinear}
J.~Huang, {\em Nonlinear output regulation: theory and applications}.
\newblock SIAM, 2004.

\bibitem{isidori2003robust}
A.~Isidori, L.~Marconi, and A.~Serrani, {\em Robust autonomous guidance: an internal model approach}.
\newblock Springer Science \& Business Media, 2003.

\bibitem{bin2020approximate}
M.~Bin, P.~Bernard, and L.~Marconi, ``Approximate nonlinear regulation via identification-based adaptive internal models,'' {\em IEEE Transactions on Automatic Control}, vol.~66, no.~8, pp.~3534--3549, 2020.

\bibitem{broucke2022adaptive}
M.~E. Broucke {\em et~al.}, ``Adaptive internal models in neuroscience,'' {\em Foundations and Trends{\textregistered} in Systems and Control}, vol.~9, no.~4, pp.~365--550, 2022.

\bibitem{su2021adaptive}
Z.~Su, A.~H. Chow, and R.~Zhong, ``Adaptive network traffic control with an integrated model-based and data-driven approach and a decentralised solution method,'' {\em Transportation Research Part C: Emerging Technologies}, vol.~128, p.~103154, 2021.

\bibitem{francis1976internal}
B.~A. Francis and W.~M. Wonham, ``The internal model principle of control theory,'' {\em Automatica}, vol.~12, no.~5, pp.~457--465, 1976.

\bibitem{francis1977linear}
B.~A. Francis, ``The linear multivariable regulator problem,'' {\em SIAM Journal on Control and Optimization}, vol.~15, no.~3, pp.~486--505, 1977.

\bibitem{davison1976robust}
E.~Davison, ``The robust control of a servomechanism problem for linear time-invariant multivariable systems,'' {\em IEEE Transactions on Automatic Control}, vol.~21, no.~1, pp.~25--34, 1976.

\bibitem{huang1990nonlinear}
J.~Huang and W.~J. Rugh, ``On a nonlinear multivariable servomechanism problem,'' {\em Automatica}, vol.~26, no.~6, pp.~963--972, 1990.

\bibitem{huang1994robust}
J.~Huang and C.-F. Lin, ``On a robust nonlinear servomechanism problem,'' {\em IEEE Transactions on Automatic Control}, vol.~39, no.~7, pp.~1510--1513, 1994.

\bibitem{huang1995asymptotic}
J.~Huang, ``Asymptotic tracking and disturbance rejection in uncertain nonlinear systems,'' {\em IEEE Transactions on Automatic Control}, vol.~40, no.~6, pp.~1118--1122, 1995.

\bibitem{huang2004general}
J.~Huang and Z.~Chen, ``A general framework for tackling the output regulation problem,'' {\em IEEE Transactions on Automatic Control}, vol.~49, no.~12, pp.~2203--2218, 2004.

\bibitem{byrnes2003limit}
C.~I. Byrnes and A.~Isidori, ``Limit sets, zero dynamics, and internal models in the problem of nonlinear output regulation,'' {\em IEEE Transactions on Automatic Control}, vol.~48, no.~10, pp.~1712--1723, 2003.

\bibitem{byrnes2004nonlinear}
C.~I. Byrnes and A.~Isidori, ``Nonlinear internal models for output regulation,'' {\em IEEE Transactions on Automatic Control}, vol.~49, no.~12, pp.~2244--2247, 2004.

\bibitem{huang2001remarks}
J.~Huang, ``Remarks on the robust output regulation problem for nonlinear systems,'' {\em IEEE Transactions on Automatic Control}, vol.~46, no.~12, pp.~2028--2031, 2001.

\bibitem{byrnes1997structurally}
C.~I. Byrnes, F.~D. Priscoli, A.~Isidori, and W.~Kang, ``Structurally stable output regulation of nonlinear systems,'' {\em Automatica}, vol.~33, no.~3, pp.~369--385, 1997.

\bibitem{nikiforov1998adaptive}
V.~O. Nikiforov, ``Adaptive non-linear tracking with complete compensation of unknown disturbances,'' {\em European journal of control}, vol.~4, no.~2, pp.~132--139, 1998.

\bibitem{marino2003output}
R.~Marino and P.~Tomei, ``Output regulation for linear systems via adaptive internal model,'' {\em IEEE Transactions on Automatic Control}, vol.~48, no.~12, pp.~2199--2202, 2003.

\bibitem{basturk2015adaptive}
H.~I. Basturk and M.~Krstic, ``Adaptive sinusoidal disturbance cancellation for unknown {LTI} systems despite input delay,'' {\em Automatica}, vol.~58, pp.~131--138, 2015.

\bibitem{serrani2001semi}
A.~Serrani, A.~Isidori, and L.~Marconi, ``Semi-global nonlinear output regulation with adaptive internal model,'' {\em IEEE Transactions on Automatic Control}, vol.~46, no.~8, pp.~1178--1194, 2001.

\bibitem{liu2009parameter}
L.~Liu, Z.~Chen, and J.~Huang, ``Parameter convergence and minimal internal model with an adaptive output regulation problem,'' {\em Automatica}, vol.~45, no.~5, pp.~1306--1311, 2009.

\bibitem{li2012nonlinear}
R.~Li and H.~K. Khalil, ``Nonlinear output regulation with adaptive conditional servocompensator,'' {\em Automatica}, vol.~48, no.~10, pp.~2550--2559, 2012.

\bibitem{lu2019adaptive}
M.~Lu, L.~Liu, and G.~Feng, ``Adaptive tracking control of uncertain euler--lagrange systems subject to external disturbances,'' {\em Automatica}, vol.~104, pp.~207--219, 2019.

\bibitem{marconi2008uniform}
L.~Marconi and L.~Praly, ``Uniform practical nonlinear output regulation,'' {\em IEEE Transactions on Automatic Control}, vol.~53, no.~5, pp.~1184--1202, 2008.

\bibitem{isidori2012robust}
A.~Isidori, L.~Marconi, and L.~Praly, ``Robust design of nonlinear internal models without adaptation,'' {\em Automatica}, vol.~48, no.~10, pp.~2409--2419, 2012.

\bibitem{wang2023nonparametric}
S.~Wang, M.~Guay, Z.~Chen, and R.~D. Braatz, ``A nonparametric learning framework for nonlinear robust output regulation,'' {\em IEEE Transactions on Automatic Control}, DOI: 10.1109/TAC.2024.3470065, 2024.

\bibitem{wang2024nonparametric}
S.~Wang, M.~Guay, and R.~D. Braatz, ``Nonparametric steady-state learning for robust output regulation of nonlinear output feedback systems,'' {\em arXiv preprint arXiv:2402.16170}, 2024.

\bibitem{afri2016state}
C.~Afri, V.~Andrieu, L.~Bako, and P.~Dufour, ``State and parameter estimation: {A} nonlinear {Luenberger} observer approach,'' {\em IEEE Transactions on Automatic Control}, vol.~62, no.~2, pp.~973--980, 2016.

\bibitem{zhou2006new}
B.~Zhou and G.-R. Duan, ``A new solution to the generalized {Sylvester} matrix equation {AV-EVF= BW},'' {\em Systems \& Control Letters}, vol.~55, no.~3, pp.~193--198, 2006.

\bibitem{marconi2007output}
L.~Marconi, L.~Praly, and A.~Isidori, ``Output stabilization via nonlinear luenberger observers,'' {\em SIAM Journal on Control and Optimization}, vol.~45, no.~6, pp.~2277--2298, 2007.

\bibitem{KE2003}
G.~Kreisselmeier and R.~Engel, ``Nonlinear observers for autonomous {Lipschitz} continuous systems,'' {\em IEEE Transactions on Automatic Control}, vol.~48, no.~3, pp.~451--464, 2003.

\bibitem{chen2015stabilization}
Z.~Chen and J.~Huang, {\em Stabilization and regulation of nonlinear systems}.
\newblock Springer, 2015.

\bibitem{sontag1995changing}
E.~Sontag and A.~Teel, ``Changing supply functions in input/state stable systems,'' {\em IEEE Transactions on Automatic Control}, vol.~40, no.~8, pp.~1476--1478, 1995.

\bibitem{su2013cooperative}
Y.~Su and J.~Huang, ``Cooperative adaptive output regulation for a class of nonlinear uncertain multi-agent systems with unknown leader,'' {\em Systems \& Control Letters}, vol.~62, no.~6, pp.~461--467, 2013.

\end{thebibliography}
\end{document}